\documentclass[conference]{IEEEtran}
\IEEEoverridecommandlockouts
\usepackage{cite}
\usepackage{amsmath,amssymb,amsfonts}
\usepackage{graphicx}
\usepackage{textcomp}
\usepackage{tikz}
\usepackage{xcolor}

\usetikzlibrary{arrows.meta,calc}

\usepackage[colorlinks]{hyperref}
\hypersetup{
     colorlinks=true,
     linkcolor=blue,
     filecolor=blue,
     citecolor = blue,
     urlcolor=cyan,
     }

\usepackage{amsthm}

\usepackage{booktabs}
\usepackage{caption}
\usepackage{cite}
\usepackage{colortbl}
\usepackage{pdfpages}
\usepackage{dsfont}

\usepackage{bm}
\usepackage{bbm}
\usepackage{dirtytalk}

\usepackage{subcaption}

\newtheorem{theorem}{Theorem}

\newtheorem{lemma}{Lemma}

\newtheorem{proposition}{Proposition}

\theoremstyle{definition}

\newtheorem{definition}{Definition}

\theoremstyle{definition}

\theoremstyle{definition}

\definecolor{DarkGreen}{rgb}{0.1,0.5,0.1}
\definecolor{DarkRed}{rgb}{0.5,0.1,0.1}
\definecolor{DarkBlue}{rgb}{0.1,0.1,0.5}
\definecolor{DarkPurple}{rgb}{0.5,0.2,0.5}
\definecolor{DarkTurquoise}{rgb}{0.1,0.5,0.5}

\DeclareMathOperator{\tr}{tr}

\usepackage{booktabs}
\usepackage{diagbox}

\def\BibTeX{{\rm B\kern-.05em{\sc i\kern-.025em b}\kern-.08em
    T\kern-.1667em\lower.7ex\hbox{E}\kern-.125emX}}
\begin{document}

\title{Field Trace Polynomial Codes for Secure Distributed Matrix Multiplication\\
}

\author{Roberto Assis Machado\IEEEauthorrefmark{1}, Rafael G. L. D'Oliveira\IEEEauthorrefmark{2}, Salim El Rouayheb\IEEEauthorrefmark{3} and Daniel Heinlein\IEEEauthorrefmark{4}   \\ \IEEEauthorrefmark{1}SMSS, Clemson University, USA \\ \IEEEauthorrefmark{2}RLE, Massachusetts Institute of Technology, USA\\ \IEEEauthorrefmark{3}ECE, Rutgers University, USA \\\IEEEauthorrefmark{4}Department of Communications and Networking, Aalto University, Finland \\ Emails: robertoassismachado@gmail.com, rafaeld@mit.edu, \\ salim.elrouayheb@rutgers.edu, daniel.heinlein@aalto.fi }

\maketitle

\begin{abstract}
We consider the problem of communication efficient secure distributed matrix multiplication. The previous literature has focused on reducing the number of servers as a proxy for minimizing communication costs. The intuition being, that the more servers used, the higher the communication cost. We show that this is not the case. Our central technique relies on adapting results from the literature on repairing Reed-Solomon codes where instead of downloading the whole of the computing task, a user downloads field traces of these computations. We present field trace polynomial codes, a family of codes, that explore this technique and characterize regimes for which our codes outperform the existing codes in the literature. 
\end{abstract}
\begin{IEEEkeywords}
security, distributed computation, coding theory
\end{IEEEkeywords}
\vspace{-10pt}

\section{Introduction}

We consider the problem of secure distributed matrix multiplication (SDMM) in which a user has two matrices, $A \in \mathbb{F}_q^{a \times b}$ and $B \in \mathbb{F}_q^{b \times c}$, and wishes to compute their product, $AB \in \mathbb{F}_q^{a \times c}$, with the assistance of $N$ servers, without leaking any information about either $A$ or $B$ to any server.  We assume that all servers are honest but curious (passive), in that they are not malicious and will faithfully follow the pre-agreed upon protocol. However, any $T$ of them may collude to try to eavesdrop and deduce information about either $A$ or $B$. 

We follow the  setting proposed in~\cite{ravi2018mmult}, with many follow-up works~\cite{Kakar2019OnTC,koreans,d2019gasp,DOliveira2019DegreeTF, Aliasgari2019DistributedAP,aliasgari2020private,Kakar2019UplinkDownlinkTI,doliveira2020notes,Yu2020EntangledPC,mital2020secure,bitar2021adaptive, hasircioglu2021speeding}. The performance metric initially used was the download cost, i.e., the total amount of data downloaded by the users from the server. Subsequent work has also considered the upload cost \cite{mital2020secure}, the total communication cost \cite{9004505}, and computational costs \cite{doliveira2020notes}.

Different partitionings of the matrices lead to different trade-offs between upload and download costs. In this paper, we consider the inner product partitioning given by $A = \begin{bmatrix}A_1 & \cdots & A_L\end{bmatrix}$ and $B^{\intercal} = \begin{bmatrix}B_1^{\intercal} & \cdots & B_L^{\intercal}\end{bmatrix}$ such that $AB = A_1 B_1 + \cdots + A_L B_L$, where all products $A_\ell B_\ell$ are well-defined and of the same size. Under this partitioning, a polynomial code is a polynomial $h(x)=f(x) \cdot g(x)$, whose coefficients encode the submatrices $A_kB_\ell$. The $N$ servers compute the evaluations $h(\alpha_1),\ldots,h(\alpha_N)$ for certain $\alpha_1,\ldots,\alpha_N$. The next step is where the scheme we propose differs from previous works. In previous works, the servers send these evaluations to the user. The polynomial $h(x)$ is constructed so that no $T$-subset of evaluations reveals any information about $A$ or $B$ ($T$-security), but so that the user can reconstruct $AB$ given all $N$ evaluations (decodability). In order to contrast them with our approach, we refer to these schemes as traditional polynomial schemes.

Examples of traditional polynomial schemes for the inner product partitioning are the secure MatDot codes in \cite{Aliasgari2019DistributedAP} and the codes in~\cite{mital2020secure}. The main focus in the literature was on minimizing the minimum amount of helping servers $N$, also known as the recovery threshold, in order to reduce the communication cost. The intuition being, that the more servers used, the higher the communication cost. We show that this is not generally the case, i.e., in some cases it is possible to reduce the total communication by contacting more servers.

In this paper, we present field trace polynomial (FTP) codes, a non-traditional polynomial scheme inspired by techniques from the repair literature, specifically the trace-based methods for repairing Reed-Solomon codes, first introduced in \cite{guruswami2017repairing}. In the Reed-Solomon codes repair setting, servers store different evaluations of a polynomial $h(\alpha_i)$, some of which may be lost due to node failures. The repair problem consists of finding schemes to recover any lost evaluation while minimizing the download bandwidth (referred to as the repair bandwidth in that setting) \cite{dimakis2011survey}. The key tool utilized  is the field trace. If $\mathbb{E}$ is a finite algebraic extension of a field $\mathbb{F}$, then the field trace $\tr_{\mathbb{E}/\mathbb{F}} : \mathbb{E} \rightarrow \mathbb{F}$ is a linear functional over the extension field $\mathbb{E}$ when seen as a vector space over $\mathbb{F}$. Instead of sending full evaluations of $h(\alpha_i) \in \mathbb{E}$, the servers can repair an evaluation by sending traces $\tr_{\mathbb{E}/\mathbb{F}} (h(\alpha_i)) \in \mathbb{F}$. Since $\mathbb{F}$ is a sub-field of $\mathbb{E}$, this results in a smaller download cost. Indeed, the download cost decreases as more servers are used.  

FTP codes follow the same idea. By increasing the number of servers and transmitting traces of the evaluations, instead of the whole evaluations, FTP codes can obtain better download costs. However, since the SDMM setting is a computation-offload setting and not just an information storage one, we must also account for the upload cost. Thus, care must be taken so that the decrease in the download cost, obtained from the use of the field trace, is not outgrown by the increase in the upload cost, from having to contact more servers. In Theorem~\ref{theo:scheme}, we characterize the total communication rate achieved by FTP codes.

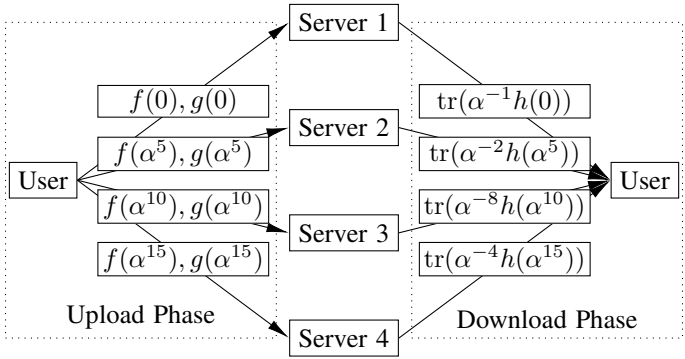
\begin{figure}[!t]
\centering
\begin{tikzpicture}
\tikzset{
myline/.style={-{Latex[length=3mm, width=1.3mm]}},
mynode/.style={pos=0.5,fill=white,draw,text width=22mm,align=center,inner sep=1pt}
}

\node[draw] (U1) at (0,0) {User};
\node[draw] (S1) at (4,2.1) {Server 1};
\node[draw] (S2) at (4,0.7) {Server 2};
\node[draw] (S3) at (4,-0.7) {Server 3};
\node[draw] (S4) at (4,-2.1) {Server 4};
\node[draw] (U2) at (8,0) {User};
\coordinate (BL1) at ($(U1 |- S4)+(-0.5,0)$);
\coordinate (TR1) at ($(S1)+(-0.9,0)$);
\coordinate (BL2) at ($(S4)+(0.9,0)$);
\coordinate (TR2) at ($(U2 |- S1)+(0.5,0)$);

\draw[myline] (U1.east) -- (S1.west) node[mynode] {$f(0), g(0)$};
\draw[myline] (U1.east) -- (S4.west) node[mynode] {$f(\alpha^{15}), g(\alpha^{15})$};
\draw[myline] (U1.east) -- (S2.west) node[mynode] {$f(\alpha^{5}), g(\alpha^{5})$};
\draw[myline] (U1.east) -- (S3.west) node[mynode] {$f(\alpha^{10}), g(\alpha^{10})$};

\draw[myline] (S1.east) -- (U2.west) node[mynode] {$\tr(\alpha^{-1}h(0))$};
\draw[myline] (S2.east) -- (U2.west) node[mynode] {$\tr(\alpha^{-2}h(\alpha^5))$};
\draw[myline] (S4.east) -- (U2.west) node[mynode] {$\tr(\alpha^{-4}h(\alpha^{15}))$};
\draw[myline] (S3.east) -- (U2.west) node[mynode] {$\tr(\alpha^{-8}h(\alpha^{10}))$};

\draw[dotted] (BL1) rectangle (TR1);
\draw[dotted] (BL2) rectangle (TR2);
\node[above] at ($(BL1)!0.5!(BL1 -| TR1)$) {Upload Phase};
\node[above] at ($(BL2)!0.5!(BL2 -| TR2)$) {Download Phase};
\end{tikzpicture}
\caption{An example of an FTP code detailed in Section \ref{sec: motivating example}. The user computes carefully chosen evaluations of the polynomials $f(x)$ and $g(x)$ and uploads them to the servers. Each server then computes the product of their received evaluations, which is itself an evaluation of the polynomial $h(x) = f(x) \cdot g(x)$. In a traditional polynomial scheme, the servers would then transmit these evaluations to the user. Utilizing an FTP code, the servers compute the trace of theses evaluations and sends them instead. Thus, decreasing the download cost.}\label{fig:phase1and2}
\end{figure}

\begin{theorem}\label{theo:scheme}
Let $L$ and $T$ be positive integers, $p_1, \ldots, p_L$ be prime numbers, listed in increasing order, $q_0$ be a prime power, and set $q=q_0^{p_1p_2\ldots p_L}$. Let $A \in \mathbb{F}_q^{a \times b}$, $B \in \mathbb{F}_q^{b \times c}$ be two matrices, and $N_i = p_i + 2L + 2T - 2$, for every $i \in [L]$.
Then, there exists an FTP code with partitioning parameter $L$ and security parameter $T$, which securely computes $AB$ utilizing $N_L$ servers with a total communication rate of
\begin{align} \label{eq: FTP rate}
\mathcal{R} = \left(\frac{N_Lb}{L} \left( \frac{1}{a}+\frac{1}{c} \right)+\sum_{i=1}^L\frac{N_i}{p_i}\right)^{-1} .
\end{align}
\end{theorem}

In Theorem~\ref{theo:comparison}, we show that FTP codes outperform any traditional polynomial scheme when the number of columns in $A$ (or rows in $B$) is sufficiently smaller than both the number of rows in $A$ and columns in $B$.

\begin{theorem}\label{theo:comparison}
For every traditional polynomial scheme, there exists a constant $K>0$ and an FTP code such that the communication rate of the FTP code is higher than that of the traditional polynomial scheme, whenever the matrix dimensions $a$, $b$, and $c$ are such that $b \left( \frac{1}{a}+\frac{1}{c} \right) < K$.
\end{theorem}

The expression for the constant $K$ is given in \eqref{eq: K}. Theorem~\ref{theo:comparison} is proved by constructing a particular FTP code and comparing it to the trivial lower bound on the recovery threshold of a traditional polynomial scheme, namely that $N>L$. To provide some context, the current state of the art for inner product partitioning is given by the traditional polynomial scheme in \cite{mital2020secure} which has a recovery threshold of $N=L+2T$ (or $N=L+T$ with precomputations).

\subsection{Related Work}

For distributed computations, polynomial codes were originally introduced in~\cite{polycodes1} in order to mitigate stragglers in distributed matrix multiplication. This was followed by a series of works,~\cite{polycodes2,pulkit,pulkit2,fundamental}.

The literature on SDMM has also studied different variations on the model we focus on here. For instance, in~\cite{nodehi2018limited,jia2019capacity,mital2020secure,akbari2021secure} the encoder and decoder are considered to be separate, in~\cite{nodehi2018limited} servers are allowed to cooperate, and in~\cite{kim2019private} the authors consider a hybrid between SDMM and private information retrieval where the user has a matrix $A$ and wants to privately multiply it with a matrix $B$ belonging to some public list. FTP codes can be readily used or adapted to many of these settings as has been done with other polynomial schemes (e.g.,~\cite{bitar2021adaptive, zhu2021improved}).
 
There is now a vast literature on the repair problem for distributed storage systems (e.g., \cite{dimakis2010network,dimakis2011survey,rashmi2011optimal,cadambe2013asymptotic}). The field trace method relevant to us was developed in \cite{guruswami2017repairing} and later extended in \cite{tamo2017optimal}. Methods from repair were used in \cite{bitar2020minimizing, rawat2018centralized, huang2017secret} to construct communication-efficient secret sharing. 

\subsection{Main Contributions}

Our main contributions are summarized below.

\begin{itemize}

    \item We show a connection between SDMM and the repair problem of Reed-Solomon codes for distributed storage. This is done by treating the user as a repair node wishing to restore a polynomial evaluation. The essential difference between these settings is that in SDMM the upload cost is important. This occurs because the SDMM setting is a computation-offload setting and not just an information storage one, in which the storage upload can be amortized. Other differences include more flexibility in choosing the amount of servers, in designing the code, and in choosing the evaluation points.
    
    \item By adapting the techniques used in the repair of Reed-Solomon codes, we present FTP codes for SDMM. Contrary to traditional polynomial codes, FTP codes achieve higher download rates by communicating with more servers. We show that they are secure, decodable, and present their total communication rate in Theorem~\ref{theo:scheme}.

    \item In Theorem~\ref{theo:comparison}, we show that FTP codes outperform any traditional polynomial scheme when the number of columns in $A$ (or rows in $B$) is sufficiently smaller than both the number of rows in $A$ and columns in $B$.

\end{itemize}

\section{Preliminaries}

In this section, we introduce some notation and concepts needed for the rest of the paper. For example, we define $[M,N] = \{M,M+1,\ldots,N\}$ and $[M] = [1,M]$.

\begin{definition}
Let $\mathcal{C}$ be a linear code of length $n$ over a finite field $\mathbb{F}_q$. The dual code of $\mathcal{C}$ is the linear subspace
\begin{align*}
\mathcal{C}^\perp = \left\{d\in \mathbb{F}_q^n : \sum_{i=1}^n d_i c_i = 0  \quad\forall c \in \mathcal{C}  \right\}
.
\end{align*}

\end{definition}
\begin{definition}\label{def:GRS}
Let $V = (v_1, v_2, \ldots, v_n)\in \mathbb{F}_q^n$ be a vector with non-zero entries and $\Omega = \{\alpha_1, \alpha_2 \ldots, \alpha_n\} \subset \mathbb{F}_q$ a set of distinct elements. The \emph{Generalized Reed-Solomon Code} with parameters $n,k,\Omega,$ and $V$ is given by
$GRS_{\mathbb{F}_q}(n, k, \Omega,V) = \{(v_1f(\alpha_1), v_2f(\alpha_2), \ldots, v_nf(\alpha_n)): f(x)\in \mathbb{F}_q[x] \text{ and } \deg(f) < k\}$.
\end{definition}

If $V= (1,1,\ldots, 1)$, then $GRS_{\mathbb{F}_q}(n, k, \Omega,V)$ is the Reed-Solomon Code $RS_{\mathbb{F}_q}(n, k, \Omega)$. In addition, the dual code of an $RS_{\mathbb{F}_q}(n, k, \Omega)$ is $GRS_{\mathbb{F}_q}(n, n-k, \Omega,V')$, where $V'=(v_1', \ldots, v_n')$ and $(v_i')^{-1} = \prod_{\substack{1\leq j\leq n \\ j\neq i}}(\alpha_i - \alpha_j)$.

Let $\mathbb{E}$ be a finite algebraic extension of a field $\mathbb{F}$. The degree $s=[\mathbb{E}: \mathbb{F}]$ of a field extension $\mathbb{E}/\mathbb{F}$ is the dimension of $\mathbb{E}$ as a vector space over $\mathbb{F}$. Thus, any element $v \in \mathbb{E}$ can be expressed as a vector $(v_0, v_1, \ldots, v_{s-1})\in\mathbb{F}^s$.

\begin{definition}
Let $\mathbb{E}$ be a finite algebraic extension of a field $\mathbb{F}$. The field trace $\tr_{\mathbb{E}/\mathbb{F}} : \mathbb{E} \rightarrow \mathbb{F}$ is the $\mathbb{F}$-linear functional
$\tr_{\mathbb{E}/\mathbb{F}}(x) = x + x^{|\mathbb{F}|} + x^{|\mathbb{F}|^2} +\cdots + x^{|\mathbb{F}|^{s-1}} .$
\end{definition}

Note that the codomain of the field trace is the subfield $\mathbb{F}$. It is this fact that allows savings in the download cost. We also note that, given an $\mathbb{F}$-basis $\{\lambda_0, \lambda_1, \ldots, \lambda_{s-1}\}$ of $\mathbb{E}$, there exists a \emph{trace-dual} $\mathbb{F}$-basis $\{\mu_0, \mu_1, \ldots, \mu_{s-1}\}$ of $\mathbb{E}$, i.e., such that $\tr_{\mathbb{E}/\mathbb{F}}(\lambda_i \mu_j)$ equals $1$ if $i=j$, and equals $0$ otherwise.

The next proposition plays a crucial role in proving the decodability of FTP codes in Lemma~\ref{lem:decodability}.

\begin{proposition}[Page 759 in \cite{huffman2021concise}]
Let $\{\lambda_0, \lambda_1, \ldots, \lambda_{s-1}\}$ and $\{\mu_0, \mu_1, \ldots, \mu_{s-1}\}$ be trace-dual $\mathbb{F}$-bases of $\mathbb{E}$. Then, given an element $\beta \in \mathbb{E}$, the coefficients of its expansion in the basis $\{\mu_0, \mu_1, \ldots, \mu_{s-1}\}$ are given by  $\tr_{\mathbb{E}/\mathbb{F}}(\lambda_i \beta)$, so that
\begin{equation}\label{dualexpansion}
    \beta = \sum_{i=0}^{s-1} \tr_{\mathbb{E}/\mathbb{F}}(\lambda_i \beta)\mu_i
    .
\end{equation}
\end{proposition}

\section{A Motivating Example: $L=T=1$} \label{sec: motivating example}
\label{sec3}

We begin our description of FTP codes with the following example, which we present in as much detail as possible in order to showcase the essential ingredients of the scheme. We compare the traditional polynomial scheme using $N'=3$ servers with an FTP code using $N=4$ server.

In this example, a user wishes to multiply two matrices $A \in \mathbb{F}_{16}^{a \times b}$ and $B \in \mathbb{F}_{16}^{b \times c}$ with the assistance of non-colluding helper servers. The solution to this via traditional polynomial schemes utilizes $N'=3$ servers and involves picking two random matrices $R \in \mathbb{F}_{16}^{a \times b}$ and $S \in \mathbb{F}_{16}^{b \times c}$ and constructing the polynomials $f'(x) = A + Rx$ and $g'(x) = B + Sx$. The user then selects three distinct non-zero elements $\beta_1,\beta_2,\beta_3 \in \mathbb{F}_{16}$ and uploads both $f'(\beta_i)$ and $g'(\beta_i)$ to Server $i$. Each server then computes the product $f'(\beta_i) \cdot g'(\beta_i)$. This is equivalent to computing an evaluation $h'(\beta_i)$ of the polynomial $h'(x) = AB + (AS + RB)x + RSx^2$. The user then downloads each $h'(\beta_i)$, obtaining three evaluations of a polynomial of degree two. Therefore, the user can retrieve the polynomial $h(x)$ and compute $h(0)=AB$ as desired.

Security of the traditional scheme follows from the fact that $I(f'(\beta_i), g'(\beta_i);A,B)=0$. As for the communication costs, first the user uploads $f'(\beta_i)$ and $g'(\beta_i)$, which cost $ab$ and $bc$, symbols respectively, three times. Thus, the upload cost is $3(ab+bc)$ symbols of $\mathbb{F}_{16}$. Then, the user downloads $h(\beta_i)$, which costs $ac$ symbols of $\mathbb{F}_{16}$, three times, obtaining a download cost of $3ac$ symbols of $\mathbb{F}_{16}$. Since the user retrieves $AB \in \mathbb{F}_{16}^{a \times c}$, which consists of $ac$ symbols of $\mathbb{F}_{16}$, the total communication rate is given by $\mathcal{R}' = \frac{ac}{3ab+3bc+3ac}$.

Based on techniques from the literature on repairing Reed-Solomon codes for distributed storage, we present an FTP code which obtains a lower download cost by utilizing $N=4$ servers. Let $\alpha \in \mathbb{F}_{16}$ be an algebraic element of degree $2$ such that $\alpha^4 + \alpha +1 = 0$. Then, the finite field $\mathbb{F}_2(\alpha)$, constructed by extending the binary field with the algebraic element $\alpha$, is identified by $\mathbb{F}_{16}$. Let $f(x) = A+R(x-\alpha)$ and $g(x) = B +S(x-\alpha)$, where $R$ and $S$ are the same random matrices used in the traditional polynomial scheme above. Then, $h(x) = f(x) \cdot g(x)$ is such that $h(\alpha) = AB$.\footnote{Unlike in the repairing Reed-Solomon codes setting, in the SDMM setting, we have more freedom to design the code and choose the evaluation points.}

Our scheme, illustrated in Figure \ref{fig:phase1and2}, works as follows. First, the user uploads the evaluations $f(y_i)$ and $g(y_i)$ to each Server $i$. Then, each Server $i$ computes $\tr(\alpha^{-j_i}h(y_i))$, where $(j_1,j_2,j_3, j_4) = (1,2,8,4)$ and $(y_1, y_2, y_3, y_4)= (0, \alpha^{5}, \alpha^{10}, \alpha^{15})$, and sends it back to the user.\footnote{Here, we denote the field trace by $\tr := \tr_{\mathbb{F}_{16}/\mathbb{F}_{4}}:\mathbb{F}_{16} \rightarrow \mathbb{F}_4$ given by $\tr(x) = x + x^{4}$. We apply the trace function element-wise on matrices which is equivalent to an element-wise exponentiation.} To show that the user is able to decode $AB$, we denote the answer of each Server~$i$ by $S_i$. Then, \begin{align*}
&
\,\,\,\,\,\,\,\,
\alpha^4(S_1+S_2+S_3+S_4) + \alpha^5 S_2 + \alpha^{10} S_3 + \alpha^{15} S_4
\\&=
\alpha^4 \tr(\alpha^{-1} h(0) +\alpha^{-2} h(\alpha^5) + \alpha^{-8} h(\alpha^{10})+ \alpha^{-4} h(\alpha^{15}))
\\&
\,\,\,\,\,\,\,
+\tr(\alpha^{3} h(\alpha^5) + \alpha^{2} h(\alpha^{10}) + \alpha^{11} h(\alpha^{15}))
\\&=
\alpha^4 \tr(h(\alpha))  + \tr(\alpha h(\alpha)) 
\\&=
h(\alpha) = AB
\end{align*}
The first equality follows from the $\mathbb{F}_4$-linearity of the field trace together with  the fact that $\alpha^{5},\alpha^{10},\alpha^{15} \in \mathbb{F}_4$. The second equality follows from utilizing the equation $h(x) = h_0 + h_1 x + h_2 x^2$ and the facts that $\alpha^{15} = 1$ and $\alpha^4+\alpha+1 = 0$. The third equality follows from utilizing the equation $\tr(x)=x+x^4$ and the fact that $\alpha^4+\alpha+1 = 0$.

Security follows by showing that $I(f(y_i), g(y_i);A,B)=0$, as is done in Lemma \ref{lem:tsecure}. As for the communication costs, first the user uploads $f(y_i)$ and $g(y_i)$, which cost $2ab$ and $2bc$, symbols respectively, four times. Thus, the upload cost is $4(2ab+2bc)$ symbols of $\mathbb{F}_{4}$. Then, the user downloads $\tr(\alpha^{-j_i}h(y))$, which costs $ac$ symbols of $\mathbb{F}_{16}$, four times, obtaining a download cost of $4ac$ symbols of $\mathbb{F}_{4}$. Since the user retrieves $AB \in \mathbb{F}_{16}^{a \times c}$, which consists of $2ac$ symbols of $\mathbb{F}_{4}$, the total communication rate is given by $\mathcal{R} = \frac{ac}{4ab+4bc+2ac}$. 

We note that the download cost of the FTP code is lower than that of the traditional polynomial code, and that the opposite is true for the upload cost. In terms of total communication, the FTP code outperforms the traditional one whenever $\mathcal{R}>\mathcal{R}'$. This occurs whenever the matrix dimensions satisfy the inequality $b\left( \frac{1}{a}+\frac{1}{c} \right)<1$.

We note that this FTP code could also outperform any traditional polynomial scheme with $N'=2$ servers, if such a scheme existed. Indeed, the only way for a traditional polynomial scheme to outperform the FTP code for all matrix dimensions would be for it to use $N'=1$ server, which is not possible because of the $1$-security constraint.

\section{Field Trace Polynomial Codes}

In this section, we present the general construction for FTP codes. The main idea is to perform the same technique as in Section~\ref{sec3}, $L$ times, each one retrieving $A_i B_i$, but doing so utilizing a single polynomial $h(x) = f(x) \cdot g(x)$.

\noindent \textbf{Choosing the Field:} We begin by choosing the field over which we operate. Let $L \in \mathbb{Z}^+$, $\{p_1, \ldots, p_L\}$ be a set of distinct prime numbers in increasing order, $q_0$ a prime power, and set $q=q_0^{p_1p_2\ldots p_L}$. We then operate over $\mathbb{F}_q$. For $i \in [L]$, we let $\alpha_i \in \mathbb{F}_{q_0}$ be such that $\mathbb{F}_{q_0}(\alpha_i)$ is a field extension of $\mathbb{F}_{q_0}$ of order $p_i$. 
And thus, $\mathbb{F}_q = \mathbb{F}_{q_0}(\alpha_1, \alpha_2, \ldots, \alpha_{L})$. We also define $F_i = \mathbb{F}_{q_0}(\alpha_j : 1\leq j \leq L \text{ and } j\neq i)$.

\noindent \textbf{Choosing the Polynomials:} As described in the introduction, we consider the setting where the user partitions the matrices $A \in \mathbb{F}_q^{a \times b}$ and $B \in \mathbb{F}_q^{b \times c}$ as $A = \begin{bmatrix}A_1 & \cdots & A_L\end{bmatrix}$ and as $B^{\intercal} = \begin{bmatrix}B_1^{\intercal} & \cdots & B_L^{\intercal}\end{bmatrix}$ such that $AB = A_1 B_1 + \cdots + A_L B_L$, where each $A_i\in\mathbb{F}_q^{a \times \frac{b}{L}}$ and $B_i\in\mathbb{F}_q^{\frac{b}{L}\times c}$. In order to obtain $T$-security $R_1, \ldots, R_T \in \mathbb{F}_q^{a \times \frac{b}{L}}$ and $S_1, \ldots, S_T \in \mathbb{F}_q^{\frac{b}{L}\times c}$ are chosen independently and uniformly at random. We then define $f,g \in \mathbb{F}_q[x]$ as the polynomials of degree $L+T-1$ such that, $f(\alpha_i)=A_i$, $g(\alpha_i)=B_i$, for every $i \in [L]$, and $f(\alpha_{L+i})=R_j$, $g(\alpha_{L+i})=S_j$, for every $j \in [T]$.

\noindent \textbf{Choosing the Evaluation Points:} For each $i\in [L]$, denote $N_i = p_i + 2L + 2T - 2$ and $n = N_L+L = p_L + 3L + 2T - 2$. Consider the set $\{\alpha_1, \alpha_2, \ldots, \alpha_{L}\}$ of primitive elements defined above. Let $\alpha_{L+1}, \ldots, \alpha_{n} \in \mathbb{F}_{q_0}$ be distinct elements which are also distinct from $\{\alpha_1, \alpha_2, \ldots, \alpha_{L}\}$. We then define $\Omega=\{\alpha_1, \ldots, \alpha_{n}\}$. The evaluation points the user sends are those in the set $\{\alpha_{L+1}, \ldots, \alpha_n\}$ of size $N_L$.\footnote{Thus, $q_0\geq N_L$ is required.}

\noindent \textbf{Upload Phase:} The FTP code uses $N_L$ servers. The user uploads $f(\alpha_{L+i})$ and $g(\alpha_{L+i})$ to each Server $i$.

\noindent \textbf{Download Phase:} Let $v_j = \prod_{\substack{1\leq i\leq n \\ i\neq j}}   (\alpha_j-\alpha_i)^{-1}$ and $k_i(x)$ be the annihilator polynomial for $\{\alpha_j: j\in [n]\setminus ([L+1:N_i+L] \cup\{i\})\}$. Then, for each $i\in [L]$, Server $j$ computes  $\tr_{\mathbb{F}_q/F_i} (v_{L+j}k_i(\alpha_{L+j}) h(\alpha_{L+j}))$ and sends these $F_i$-values to the user.

\noindent \textbf{User Decoding:} In Lemma~\ref{lem:decodability}, we show that the user is able to retrieve $h(\alpha_i) = A_iB_i$ from 
$\{\tr_{\mathbb{F}_q/F_i} (v_jk_i(\alpha_{j}) h(\alpha_j)): j\in [L+1:N_i+L]\}$. Combining these, the user can decode $AB = A_1B_1+ \ldots +A_LB_L$.

\section{Proof of Theorem~\ref{theo:scheme}}

We break the proof into different Lemmas. We show that FTP codes are decodable, in Lemma~\ref{lem:decodability}, $T$-secure, in Lemma~\ref{lem:tsecure}, and characterize their performance, in Lemma~\ref{lem:comcosts}. These statements combined prove Theorem~\ref{theo:scheme}. 

\begin{lemma}\label{lem:decodability}
Given positive integers $L$ and $T$, let $p_1, p_2, \ldots, p_L$ be distinct prime numbers in the ascending order and $q_0$ be a prime power with $q_0\geq N_L$.
Let $A =\begin{bmatrix}A_1 & \cdots & A_L\end{bmatrix} \in \mathbb{F}_q^{a \times b}$ and $B^{\intercal} = \begin{bmatrix}B_1^{\intercal} & \cdots & B_L^{\intercal}\end{bmatrix} \in \mathbb{F}_q^{c \times b}$, where $q=q_0^{p_1p_2\ldots p_L}$. Then, $h(\alpha_i)$ can be decoded using $N_i$ servers, for $i\in [L]$,.
\end{lemma}
\begin{proof}

Let $f(x),g(x)\in \mathbb{F}_q[x]$ be polynomials such that
\begin{align*}
f(\alpha_i)&=A_i, &g(\alpha_i)&=B_i &&&\text{ for } i \in [L] \text{ and} \\
f(\alpha_{L+i})&=R_i, &g(\alpha_{L+i})&=S_i &&&\text{ for } i \in [T],
\end{align*}
using the inner product partitioning $A = \begin{bmatrix}A_1 & \cdots & A_L\end{bmatrix}$ and $B^{\intercal} = \begin{bmatrix}B_1^{\intercal} & \cdots & B_L^{\intercal}\end{bmatrix}$ and uniformly distributed random $\mathbb{F}_q$-matrices $R_i,S_i$.
Therefore, $h(x)=f(x)g(x)$ is a polynomial of degree $2L+2T-2$ such that $h(\alpha_i)=A_iB_i$, for $i \in [L]$.

Let $\mathcal{C}=RS_{\mathbb{F}_q}(n, 2L+2T-1, \Omega)$.
Since the degree of $h$ is smaller than $2L+2T-1$, the vector $\left(h(\alpha_1), h(\alpha_2), \ldots, h(\alpha_{n})\right)$ is in $\mathcal{C}$.

Define, for $i \in [L]$, $U_i =\{L+1, \ldots, L+N_i\}$, which has $N_i$ elements.
For each $i \in [L]$, let
\begin{align*}
k_i(x) = \prod_{j\in[n]\setminus (U_i\cup\{i\})}(x-\alpha_{j}).
\end{align*}

Note that
\begin{align*}
\deg(k_i(x))
&= (p_L+3L+2T-2) - (p_i+2L+2T-1)
\\&= p_L+L-1-p_i,
\end{align*}
and so $\deg(k_i(x)x^s) < p_L+L-1$ for $s=0,\ldots, p_i-1$.

Using $V$ as in Definition~\ref{def:GRS}, it follows that the Generalized Reed-Solomon with parameter $n$, $p_L+L-1$,  $\Omega$ and $V$ contains the element $(v_1 k_i(\alpha_1)\alpha_1^s, \ldots, v_n k_i(\alpha_n)\alpha_n^s)$, for $s=0,\ldots, p_i-1$, i.e.,
\begin{align*}
&(v_1 k_i(\alpha_1)\alpha_1^s, \ldots, v_n k_i(\alpha_n)\alpha_n^s)
\\&\in \mathcal{C}^\perp
= GRS_{\mathbb{F}_q}(n, p_L+L-1,  \Omega, V).
\end{align*}

The dual-code property implies that
\begin{align*}
\sum_{j=1}^{n} v_jk_i(\alpha_{j})\alpha_{j}^s h(\alpha_{j}) = 0.
\end{align*}

For each $i \in [L]$, we have $k_i(\alpha_j)=0$ if $j \not \in U_i \cup \{i\}$, hence
\begin{align*}
v_ik_i(\alpha_i)\alpha_{i}^{s}h(\alpha_{i}) =-\sum_{j\in U_i} v_jk_i(\alpha_{j})\alpha_{j}^s h(\alpha_{j})
.
\end{align*}
Applying $\tr_i:=\tr_{\mathbb{F}_q/F_i}$ to both sides yields, using the $F_i$-linearity and $\alpha_j \in F_i$ for $j \in U_i$,
\begin{align*}
\tr_i\left(v_ik_i(\alpha_i)\alpha_{i}^{s}h(\alpha_i) \right)= -\sum_{j\in U_i}  \alpha_{j}^s \tr_i\left(v_jk_i(\alpha_{j}) h(\alpha_{j})\right).
\end{align*}


Let $\{\lambda_{s,i} = v_ik_i(\alpha_i)\alpha_i^s : 0\leq s<p_i\}$. Since $v_ik_i(\alpha_i)\neq 0$ and $\alpha_i$ is a primitive element, the set $\{\lambda_{s,i}\}$ is an $F_i$-basis of $\mathbb{E}$.
Further, there exists a set $\{\mu_{s,i}:0\leq s<p_i \}$ which is the trace-dual $\mathbb{F}_{q_0}$-basis to $\{\lambda_{s,i}\}$ of $\mathbb{F}_{q_0}(\alpha_i)$. Thus,
\begin{align*}
\sum_{s=0}^{p_i-1}  \tr_i\left( \lambda_{s,i} h(\alpha_i)\right)\mu_{s,i}
= h(\alpha_i)
= A_i B_i.
\end{align*}
\end{proof}
Decodability is then obtained by repeating the process given in Lemma \ref{lem:decodability} and summing it over $i \in [L]$,
\begin{align*}
\sum_{i=1}^L \sum_{s=0}^{p_i-1} \tr_i\left( v_ik_i(\alpha_i)\alpha_i^s h(\alpha_{i}) \right)\mu_{s,i}
= AB
.
\end{align*}

Next, we show that FTP codes are $T$-secure.
\begin{lemma}\label{lem:tsecure}
FTP codes are $T$-secure.
\end{lemma}


\begin{proof}
Since $f(x)$ is independent from $B$ and $g(x)$ is independent from $A$, proving $T$-security is equivalent to showing that $I(A;f(\alpha_{i_1}), \ldots, f(\alpha_{i_T}))=I(B;g(\alpha_{i_1}), \ldots, g(\alpha_{i_T}))=0$. We prove the claim for $f(x)$; the proof for $g(x)$ is analogous.

By Lagrange interpolation, $f(x)$ can be expressed as
\begin{align*}
f(x) = \sum_{i=1}^{L+T}f_i(x)f(\alpha_i)
,
\end{align*}
where the Lagrange basis polynomial $f_i(x)$ is given by 
\[ f_i(x) = \prod_{\substack{1\leq m \leq L+T
\\ m\neq i}} \frac{x-\alpha_m}{\alpha_i-\alpha_m}. \]

Then,
\begin{align*}
&I(A;f(\alpha_{i_1}), \ldots, f(\alpha_{i_T}))\\
=&H(f(\alpha_{i_1}), \ldots, f(\alpha_{i_T})) - H(f(\alpha_{i_1}), \ldots, f(\alpha_{i_T})|A)\\
\le & \sum_{j \in \mathcal{T}}H(f(\alpha_{j})) - H(f(\alpha_{i_1}), \ldots, f(\alpha_{i_T})|A)\\
=& \sum_{j \in \mathcal{T}}H(f(\alpha_{j})) - H(f^{(T)}(\alpha_{i_1}), \ldots, f^{(T)}(\alpha_{i_T})),\\
=& \frac{Tab}{L}\log(q_0) - H(f^{(T)}(\alpha_{i_1}), \ldots, f^{(T)}(\alpha_{i_T}))\\
\end{align*}
where $f^{(T)}(x) = \sum_{i=1}^{T}f_{L+i}(x)f(\alpha_{L+i}) = \sum_{i=1}^{T}f_{L+i}(x)R_i$.

Since the evaluation points $\{\alpha_i: i\in \mathcal{T}\}$ are all different, the following matrix has full rank.
\[
\left(\begin{matrix}
f_{L+1}(\alpha_{i_1})&f_{L+1}(\alpha_{i_2})&\cdots& f_{L+1}(\alpha_{i_T})\\
f_{L+2}(\alpha_{i_1})&f_{L+2}(\alpha_{i_2})&\cdots& f_{L+2}(\alpha_{i_T})\\
\vdots & \vdots & \ddots & \vdots\\
f_{L+T}(\alpha_{i_1})&f_{L+T}(\alpha_{i_2})&\cdots& f_{L+T}(\alpha_{i_T})\\
\end{matrix}\right)
\]
This is because the set of $f_i's$ are linearly independent and the evaluation points are different which implies that $f^{(T)}(\alpha_{i_j})$'s are uniformly distributed in the space of the matrices $M_{a\times \frac{b}{L}}(\mathbb{F}_{q_0})$. Thus, 
$H(f^{(T)}(\alpha_{i_1}), \ldots, f^{(T)}(\alpha_{i_T})) = \frac{Tab}{L}\log(q_0)$,
and therefore, $I(A;f(\alpha_{i_1}), \ldots, f(\alpha_{i_T}))= 0$.
\end{proof}

We now characterize the total communication.

\begin{lemma}\label{lem:comcosts}
FTP codes have total communication rate
\begin{align*}
\mathcal{R} = \left(\frac{N_Lb}{L} \left( \frac{1}{a}+\frac{1}{c} \right)+\sum_{i=1}^L\frac{N_i}{p_i}\right)^{-1} .
\end{align*}
\end{lemma}


\begin{proof}
All costs will be computed in $\mathbb{F}_{q_0}$-symbols.
The upload and download costs can be directly calculated as
\begin{align*}
&\mathcal{U} = N_L\left(\frac{ab}{L}+\frac{bc}{L}\right) \prod_{j=1}^{L} p_j, \\
&\mathcal{D} = ac\sum_{i=1}^{L} N_i \prod_{j \in [L] \setminus \{i\}} p_j
.
\end{align*}

Since the matrix $AB$ has $\mathcal{S} = ac\prod_{j=1}^{L} p_j$ symbols of $\mathbb{F}_{q_0}$, it follows that the total communication rate is given by $\frac{\mathcal{S}}{\mathcal{U}+\mathcal{D}}$, which simplify to the presented formula.
\end{proof}

\section{Example: $L=3$ and $T=2$}

We present an example of an FTP code for $L=3$ and $T=2$ and compare it to the current state of the art \cite{mital2020secure}.

\textbf{Choosing the field:} 
We begin by choosing a set of three prime numbers $\{5, 7, 11\}$, $q_0=27$ a prime power, and set $27^{385}=27^{5\cdot 7\cdot 11}$. Thus, we operate over $\mathbb{F}_{27^{385}}$. Let $\alpha_1, \alpha_2, \alpha_3 \in \mathbb{F}_{27}$ be such that $\mathbb{F}_{27}(\alpha_1)$, $\mathbb{F}_{27}(\alpha_2)$, and $\mathbb{F}_{27}(\alpha_3)$ are field extensions of $\mathbb{F}_{27}$ of degrees $5$, $7$ and $11$, respectively. 
Therefore, $\mathbb{F}_{27^{385}} = \mathbb{F}_{27}(\alpha_1, \alpha_2, \ldots, \alpha_{L})$. We also define $F_1 = \mathbb{F}_{27}(\alpha_2, \alpha_3) = \mathbb{F}_{27}^{77}$, $F_2 = \mathbb{F}_{27}(\alpha_1, \alpha_3) = \mathbb{F}_{27}^{55}$, and $F_3 = \mathbb{F}_{27}(\alpha_1, \alpha_2) = \mathbb{F}_{27}^{35}$.
\\
\textbf{Choosing the polynomials:} 
Since $L=3$, consider the setting where the user partitions the matrices $A \in \mathbb{F}_{27^{385}}^{a \times b}$ and $B \in \mathbb{F}_{27^{385}}^{b \times c}$ as $A = \begin{bmatrix}A_1 & A_2 & A_3\end{bmatrix}$ and $B^{\intercal} = \begin{bmatrix}B_1^{\intercal} & B_2^{\intercal} & B_3^{\intercal}\end{bmatrix}$ such that $AB = A_1 B_1 + A_2B_2 + A_3 B_3$, where each $A_i\in\mathbb{F}_{27^{385}}^{a \times \frac{b}{3}}$ and $B_i\in\mathbb{F}_{27^{385}}^{\frac{b}{3}\times c}$. In order to obtain $2$-security we choose $R_1, R_2\in \mathbb{F}_{27^{385}}^{a \times \frac{b}{3}}$ and $S_1, S_2 \in \mathbb{F}_{27^{385}}^{\frac{b}{3}\times c}$ all independently and uniformly at random. We then define $f,g \in \mathbb{F}_{27^{385}}[x]$ as the polynomials of degree $3$ such that, $f(\alpha_i)=A_i$, $g(\alpha_i)=B_i$, for every $i \in [3]$, and $f(\alpha_{3+j})=R_j$, $g(\alpha_{3+j})=S_j$, for $j\in [2]$. Hence, $h(x)=f(x) \cdot g(x)$ is a polynomial of degree $8$ such that $h(\alpha_i)=A_iB_i$ for every $i \in [3]$.


\noindent \textbf{Choosing the Evaluation Points:} Let $N_1 = 13$, $N_2=15$, $N_3=19$, and $n = 22$. Consider the set $\{\alpha_1, \alpha_2, \alpha_{3}\}$ of primitive elements defined above. Let $\alpha_{4}, \ldots, \alpha_{22} \in \mathbb{F}_{27}$ be distinct elements which are not in $\{\alpha_1, \alpha_2, \alpha_{3}\}$. Define $\Omega=\{\alpha_1, \ldots, \alpha_{22}\}$. The evaluation points the user sends are those in the set $\{\alpha_{4}, \ldots, \alpha_{22}\}$ of size $19$.

\noindent \textbf{Upload Phase:} The FTP code uses $19$ servers. The user uploads $f(\alpha_{3+i})$ and $g(\alpha_{3+i})$ to each Server $i$.

\noindent \textbf{Download Phase:} Let $v_j = \prod_{\substack{1\leq i\leq 22 \\ i\neq j}}   (\alpha_j-\alpha_i)^{-1}$. Let $U_1 = \{4, 5, \ldots, 16\}, U_2 = \{4, 5, \ldots, 18\}$ and $U_3=\{4,\ldots, 22\}$ be the index sets of cardinality $13$, $15$, and $19$, respectively.
For $i \in [3]$, define $k_i(x)$ as the annihilator polynomial for $\{\alpha_j: j\in [22]\setminus ([4:N_i+3] \cup\{i\})\}$, i.e., $k_i(x) = \prod_{j\in[22]\setminus (U_i\cup\{i\})}(x-\alpha_{j})$.

Defining $V$ as in Definition \ref{def:GRS} for the dual code, the Server $j$, for each $i\in [3]$, computes $\tr_i(v_{3+j}k_i(\alpha_{3+j}) h(\alpha_{3+j}))$ and sends these $F_i$-values to the user.

\noindent \textbf{User Decoding:} In Lemma~\ref{lem:decodability}, we show that the user is able to retrieve $h(\alpha_i) = A_iB_i$ from 
$\{\tr_i(v_jk_i(\alpha_{j}) h(\alpha_j)): j\in [4:N_i+3]\}$. Because $V$ is constructed as in Definition \ref{def:GRS} for the dual code, it follows that
\begin{align*}
\tr_i(v_1 k_i(\alpha_i)\alpha_i^s h(\alpha_i))
&= -\sum_{j\in U_i} \alpha_{j}^s \tr_i(v_jk_i(\alpha_{j}) h(\alpha_j))
,
\end{align*}
where the trace dual is $\tr_i = \tr_{\mathbb{F}_{27^{385}}/\mathbb{F}_{27}(\alpha_j : j \in [3] \setminus \{i\})}$
Hence, given $\{\tr_i(v_jk_i(\alpha_{j}) h(\alpha_j)) : j\in U_i\}$, the user can compute $\tr_i(v_1 k_i(\alpha_i)\alpha_i^s h(\alpha_i))$. Using trace-dual basis elements $\mu_{i,s}$, for $i\in [3]$, we have
$h(\alpha_i) = \sum_{s=0}^{N_i-1} \tr_i(v_i k(\alpha_i)\alpha_i^s h(\alpha_i))\mu_{i,s}.$
Summing up the results, the user can then decode the product $AB = h(\alpha_1) + h(\alpha_2) + h(\alpha_3) = A_1B_1+ A_2B_2 + A_3B_3$.

\textbf{Communication Costs:} Server $i$ has $h(\alpha_i)$ as in terms of symbols of $\mathbb{F}_{27^{385}}$, but transmits only $\tr_i(v_jk_i(\alpha_{j}) h(\alpha_j))$, symbols of $F_i$ which is one symbol of $F_i$. To retrieve $h(\alpha_i)$, only a $p_i$-th part of the information in each of $|U_i|$ servers is needed. 
Hence, to determine $h(\alpha_i)$, it is enough to download an $x$-th part of each of the $y$ servers in $U_i$ for $(i,x,y) \in \{(1,5,13),(2,7,15),(3,11,19)\}$.


In conclusion, the FTP code used to compute the product $AB$, with partitioning parameters $L=3$ and security parameter $T=2$, uses $N_3=19$ servers and achieves a download rate of $\frac{1}{6.47}$. In comparison, the state of the art traditional polynomial code in \cite{mital2020secure} has a download rate of $\frac{1}{7}$. The scheme in \cite{mital2020secure} has a total communication rate equal to $\left(7\left(\frac{ab+bc}{3ac}+1\right)\right)^{-1}$. The total communication rate of the above FTP code is $\left(19\left(\frac{bc+ab}{3ac}+\frac{2491}{385}\right)\right)^{-1}$. A direct calculation shows that the FTP code outperforms the scheme in \cite{mital2020secure} for matrices with dimensions such that $b\left(\frac{1}{a}+\frac{1}{c}\right)< \frac{306}{2695}$.
\section{Proof of Theorem~\ref{theo:comparison}}

We begin by proving a technical lemma.

\begin{lemma}\label{lem:helper}
Let $T,N,N',L,L'$ be positive integers, $\lambda,\eta \ge 0$ be real numbers,
$p_i \ge 2L(L+T-1)\eta$ be primes, for every $i \in [L]$, such that $p_i < p_{i+1}$ and $p_L > L N' / L' -2L-2T+2$, $N_i = p_i + 2L + 2T - 2$, and $\lambda < \frac{N'-L-1/\eta}{N_L/L - N'/L'}$. Then, 
\begin{align} \label{eq: helper lemma}
    \left( \frac{N_L \lambda}{L}+\sum_{i=1}^L\frac{N_i}{p_i}\right)^{-1} >
\left(N'\left(\frac{\lambda}{L'}+1\right)\right)^{-1}.
\end{align}
\end{lemma}

\begin{proof}
Note that $p_L > L N' / L' -2L-2T+2$ is equivalent to $N_L/L - N'/L' > 0$.
Next, $p_i \ge 2L(L+T-1)\eta$ implies $\sum_{i=1}^{L} \frac{2L+2T-2}{p_i} \le \frac{1}{\eta}$.
Then, the statement implies
\begin{align*}
\left( \frac{N_L}{L} - \frac{N'}{L'} \right) \lambda
&<
N'-L-  \frac{1}{\eta}
\\&
\le
N'-L-\sum_{i=1}^{L} \frac{2L+2T-2}{p_i}
\\&
=
N'-\sum_{i=1}^{L} \frac{N_i}{p_i}.
\end{align*}
This is equivalent to $ \frac{N_L \lambda}{L} + \sum_{i=1}^{L} \frac{N_i}{p_i} < \frac{N' \lambda}{L'} + N'$, which in turn proves the claim.
\end{proof}

Consider a traditional polynomial scheme with partitioning parameter $L$ and security parameter $T$. Given the trivial lower bound $N'>L$ on its recovery threshold $N'$, it follows from a direct calculation that the total communication rate $\mathcal{R}'$ of the traditional polynomial scheme is upper bounded by 
\begin{align} \label{eq: traditional poly}
    \mathcal{R}' < \left(N'\left(\frac{b}{L}\left( \frac{1}{a}+\frac{1}{c} \right)+1\right) \right)^{-1}.
\end{align}

Let $\lambda = b \left( \frac{1}{a}+\frac{1}{c} \right)$, $\mu > (N'-L)^{-1}$ and select $L$ primes $p_1,\ldots,p_L$ in increasing order and such that $p_i \ge 2L(L+T-1)\eta$ and $p_L > L N' / L -2L-2T+2$. Next, define $N_i = p_i + 2L + 2T - 2$  and set
\begin{align} \label{eq: K}
    K=\frac{N'-L-1/\eta}{N_L/L - N'/L}.
\end{align}
Then, by Theorem~\ref{theo:scheme}, the communication rate $\mathcal{R}$ of the FTP code with these parameters, given by \eqref{eq: FTP rate}, is equal to the left hand side of \eqref{eq: helper lemma}. For these same parameters, the right hand sides of \eqref{eq: helper lemma} and \eqref{eq: traditional poly} are the same. Thus, if $b \left( \frac{1}{a}+\frac{1}{c} \right) < K$, it follows from Lemma \ref{lem:helper} that the communication rate of the FTP code $\mathcal{R}$ is larger than the right hand side of \eqref{eq: helper lemma}, and is therefore larger than the communication rate of the traditional polynomial scheme $\mathcal{R}'$.

\section*{Acknowledgment}
Rafael D'Oliveira was supported by MIT Portugal Program (Project SNOB5G with Nr. 045929 [CENTRO-01-0247-FEDER-045929]), MIT Lincoln Laboratory Purchase Order 7000500173, and the FinTech@CSAIL Research Initiative. Salim El Rouayheb was partially supported  by the NSF under grant CNS-1801630. Daniel Heinlein is supported by the Academy of Finland, Grant~331044.


\bibliographystyle{IEEEtran}
\bibliography{references.bib}

\end{document}